\documentclass[a4paper,11pt]{article}
\usepackage{amsthm,amsmath,fullpage,amsfonts,setspace,amssymb,amsthm,enumitem,algorithm2e,placeins,biblatex, graphicx,subfigure, fancyhdr,appendix, soul, listings}
\title{Randomly Generated Subgroups of the Symmetric Group and Random Lifts of Graphs}
\author{Shashwat Silas \\ \sl{University of Cambridge, UK}}
\begin{document}
\onehalfspacing
\newtheorem{thm}{Theorem}[section]
\newtheorem{defi}[thm]{Definition}
\newtheorem{lem}[thm]{Lemma}
\newtheorem{ex}[thm]{Example}
\newtheorem{cor}[thm]{Corollary}
\newtheorem{prop}[thm]{Proposition}
\newtheorem{fact}[thm]{Fact}
\newtheorem{claim}[thm]{Claim}
\newtheorem{conj}[thm]{Conjecture}
\date{\small{\today}}
\maketitle

\begin{abstract}
	Amit and Linial showed that a random lift of a graph with minimum degree $\delta\ge3$ is asymptotically almost surely $\delta$-connected, and mentioned the problem of estimating this probability as a function of the degree of the lift. We relate a randomly generated subgroup of the symmetric group on $n$ elements to random $n$-lifts of a graph and use it to provide such an estimate along with related results. We also improve their later result showing a lower bound on the edge expansion on random lifts. Our proofs rely on new ideas from group theory which make several improvements possible. We exactly calculate the probability that a random lift of a connected graph with Betti number $l$ is connected by showing that it is equal to the probability that a subgroup of the symmetric group generated by $l$ random elements is transitive. We also calculate the probability that a subgroup of a wreath product of symmetric groups generated by $l$ random generators is transitive. We show the existence of homotopy invariants in random covering graphs which reduces some of their properties to those of random regular multigraphs, and in particular makes it possible to compute the exact probability with which random regular multigraphs are connected. All our results about random lifts easily extend to iterated random lifts.

\end{abstract}
\section{Introduction}

Amit and Linial introduced random lifts of graphs in \cite{1} and studied their connectivity properties. Properties of these graphs have been widely researched since their work. Recall that a random $n$-lift, $\tilde{G}$, of a graph $G$ is constructed in the following way: arbitrarily orient the edges of $G$ and assign a permutation from the symmetric group, $\mathcal{S}_n$, to every edge of $G$ uniformly at random. The vertices of $\tilde{G}$ are $(v,i)$ where $i \in [1,n]$ for every $v \in V(G)$ and $(u,i)$ is connected to $(v,j)$ if and only if there is an edge $e$ connecting $u$ to $v$ in $G$, and this edge was assigned a permutation $\pi$ such that $\pi(i)=j$. The edges of $\tilde{G}$ are unoriented. 
	
We call $n$ the degree of the lift and $G$ the base graph. The $n$ vertices $(v,i)$ of $\tilde{G}$ form the fiber of $v \in V(G)$. Similarly, if $e$ connects $u$ to $v$ in $G$, and it is assigned the permutation $\pi$ in the construction of $\tilde{G}$, then the $n$ edges connecting $(u,i)$ to $(v,\pi(i))$ form the fiber of $e$.

If $G$ has parallel edges or loops, then we simply assign each parallel edge or loop a random permutation and construct a random lift in the same way. Lifts of graphs cover the base graph in the sense of covering spaces in topology, and there is a surjective $n$ to $1$ graph homomorphism, or covering homomorphism, from an $n$-lift of $G$ to $G$. One can also show that any walk in $G$ starting at $u$ is covered by $n$ edge-disjoint walks in any $n$-lift of $G$ (formally, each of these is a preimage of the walk in the covering homomorphism), one starting at each point in the fiber of $u$, this is known as the walk-lifting property. 

We study connectivity properties of random lifts and some of our proofs are inspired by the work of \cite{1} and \cite{2}, but we use ideas from group theory. The main theorem in \cite{1} is that asymptotically almost surely (with probability going to $1$ as $n\to \infty$), a random $n$-lift of a simple connected graph with minimum degree $\delta \ge 3$ is $\delta$-connected. Amit and Linial raise the question as to whether this probability can be estimated as a function of $n$ and suggest the study of iterated random lifts. We establish a new relationship between a random $n$-lift and a randomly generated subgroup of $\mathcal{S}_n$, which we call the walk-subgroup, and use properties of this subgroup to provide a solution to their question. We improve a result about the edge expansion of random lifts from \cite{2}, and show a new bound on the probability of $\delta$-connectivity in $n$-lifts of graphs where $\delta$ is not a fixed constant. Our results naturally extend to iterated random lifts. Interestingly, they also show the existence of properties of random lifts whose probability only depends on the homotopy type of the base graph. The main contribution of this work is methodological and we think that techniques similar to the ones in this paper could find further applications in the study of random lifts. 
\section{Outline}
 The necessary preliminaries are presented in Section 3, and the main results are in Section 4. In Section 3.1 we define general walk-subgroups. Section 3.2--3.3 describe random lifts, iterated random lifts and their respective walk-subgroups. Our proof strategy in Section 4 is to show that properties of the walk-subgroup imply properties of random lifts, and in Section 3.4 we mention results about randomly generated subgroups of $\mathcal{S}_n$ which will be useful later. 

In Section 4.1 we prove results pertaining to connectivity and edge expansion of random lifts. Amit and Linial have shown the following.
\begin{thm}[Theorem 1 \cite{1}]
Let $G$ be a simple connected graph with minimum degree $\delta\ge 3$. Then with probability $1-o_n(1)$, a random n-lift of G is $\delta$-connected. 
\end{thm}
They ask whether this probability can be estimated as a function of $n$. We first compute the probability of connectivity in Theorem 4.1, and then show how to compute a lower bound on the probability of $\delta$-connectivity in Theorem 4.7 provided that $\delta \ge 5$. In Theorem 4.10 we show that if $\delta\ge 5$, random $n$-lifts of all graphs with $k$ vertices are a.a.s. $\delta$-connected even if we assume that $\delta$ grows slowly as a function of $n$ and $k$.

In \cite{2} it is shown that the edge expansion of random lifts can be lower bounded as a function of the base graph in the following way. 

\begin{thm}[Theorem 2.1 \cite{2}]
Let G be a connected graph with $|E| >|V|$. Then there is a positive constant $\xi_0(G)$, such that a.a.s. a random lift of $G$ has expansion at least $\xi_0(G)$.
\end{thm}
We improve this result in Theorem 4.5 by showing that a slightly better lower bound holds and explicitly giving its probability.

We also show how to extend our results to iterated random lifts. The iterated construction of random lifts of random lifts of and so on, has been mentioned in \cite{1} but no properties of this model were known previously.

In Section 4.2 we show the existence of homotopy invariants in random lifts in the following sense: the probability that random lifts of $G$ have certain properties (like connectivity) depends only on the homotopy type of $G$. This result in particular extends the current understanding of how random lifts inherit structure from their base graph.

In Section 4.3 we generalize a result in \cite{4} which calculates the probability with which random elements of $\mathcal{S}_n$ generate a transitive subgroup of $\mathcal{S}_n$, to the wreath product of symmetric groups. This result may be of independent interest.

\section{Preliminaries}
\subsection{The Walk-Subgroup}
\begin{defi}
	Let $H$ be a graph and $\mathcal{G}$ a group. To every edge of $H$ associate an element of $\mathcal{G}$ through a map $V:E(G) \to \mathcal{G}$. We calculate the walk-product of the walk $\{w_1,w_2, \dots, w_n\}$ on H as $V(w_1)V(w_2) \dots V(w_n)$, where if $w_i = w_j^{-1}$ then $V(w_i) = V(w_j)^{-1}$. The subset of $\mathcal{G}$ which can be produced by walk-products is called the walk-subset of (H,V). In special cases, this subset is a subgroup of $\mathcal{G}$, which we will call the walk-subgroup of (H,V).  
\end{defi}

Walk-subsets depend on the assignment $f:E(H)\to \mathcal{G}$ and the graph $H$. For example, if $f$ assigns the identity element to every edge, then for every group $\mathcal{G}$ and graph $H$ the walk-subgroup is trivial. To see the dependence on the structure of $H$, suppose that $H$ is the path graph with group element $g_i$ on edge $i$: the walk-subset consists of the ${n \choose 2}$ elements $\Pi_{k\le i\le j} g_i$ where $1 \le k,j \le n$. 

We have the following theorem about the structure of the walk-subsets for certain special assignments. 

\begin{prop}
	Given a graph H and a spanning tree T, define $f$ to be the assignment which sends every edge in $T$ to the identity element, and any edge not in $T$ to an element $g_i \in \mathcal{G}$. In this case the walk-subgroup is the subgroup of $\mathcal{G}$ generated by the $g_i$.
\end{prop}
\begin{proof}
	This results from the following observation: we may choose a cycle basis for $H$ such that every edge of $H$ not in $T$ is in its own fundamental cycle. To generate any given element of $\langle g_1,\dots, g_i \rangle$ simply consider the walk that starts in the fundamental cycle of the first required generator, does the requisite number of loops (raising this generator to the required power), and then traverses edges in $T$ (which are all assigned the identity element) to the next required generator and so on. Finally, noting that fundamental cycles can be traversed in either direction regardless of the point of entry completes the proof. \end{proof}

\subsection{Random Lifts of Graphs}

A lift of a graph is a covering space of a graph in the topological sense. In fact, it is shown in \cite{6} that any lift of given graph can be obtained through the construction using an assignment of permutations (or perfect matchings) to edges as shown in Section 1. One may even assume that any given set of edges which does not contain a cycle is assigned the identity permutation and still obtain every lift of a graph. In random lifts these permutations are chosen independently at random. 

\begin{defi}
	A graphical property only depends on the isomorphism type of a graph. In particular, a set $\mathcal{C}$ of all graphs with graphical property $\mathcal{P}$ is a union of complete isomorphism classes of graphs.
\end{defi}

It is also well known that given any set of edges of $G$ which does not contain a cycle, the assumption that this set is assigned the identity permutation does not change the probability of any graphical property of a random $n$-lift of $G$. These edges are called \emph{flat} edges. In a way, the usual construction of random lifts has a lot of redundancy, and we can still precisely describe the graphical properties of random $n$-lifts after conditioning on assuming a subset of a subtree to be flat. For our purposes we will always work with the following assumption: given a graph $G$, we choose a spanning tree $T$ of $G$ and assume that every edge in it is assigned the identity permutation, i.e. is flat.

Since we may assume that random $n$-lifts of graphs of $G$ are constructed by assigning permutations from $\mathcal{S}_n$ uniformly at random to edges outside a spanning tree $T$, we can use Corollary 3.3 to define the walk-subgroup of random $n$-lifts of $G$ as the subgroup of $\mathcal{S}_n$ generated by $l$ random permutations where $l = |E(G)| - |V(G)| + 1$ is the number of edges of $G$ which lie outside of the spanning tree which is also known as the Betti number of the graph.

The following definitions will be useful, 
\begin{defi}[Section]
	Every vertex in a lift $\tilde{G}$ of a graph $G$ is labeled by a vertex of $G$ and an element of $[1,n]$. All vertices labeled by the same element of $[1,n]$ are collectively referred to as a section of $\tilde{G}$.
\end{defi}

\begin{defi}[Associated Walk]
	Every element $\sigma$ in the walk-subgroup of a lift is the product of permutations assigned to edges along a (not unique) walk in the base graph. Such a walk is called an associated walk of $\sigma$. 
\end{defi}
\subsection{Iterated Random Lifts}
If $G_k \to G_{k-1} \to \dots \to G_1 \to G$ is a sequence of lifts of degree $n_k, \dots, n_{1}$ respectively, then $G_k$ is an $n_k n_{k-1}\dots n_{1}$-lift of $G$. However a random lift of degree $n_2$ of a random lift of degree $n_1$ of $G$ is not a random lift of $G$. That is to say it is distributed differently than a lift produced by a randomly assigning elements of $\mathcal{S}_{n_1n_2}$ to edges of $G$. 

Iterated lifts were studied in Chapter 3.3 of \cite{7} using wreath products of symmetric groups. It is shown there that similar to the way in which $n$-lifts of graphs can be described by assigning elements of $\mathcal{S}_n$ to edges of a graph, iterated $n_k n_{k-1}\dots n_{1}$-lifts can be described by assigning elements of $\mathcal{S}_{n_k} \wr \mathcal{S}_{n_{k-1}}\wr \dots \wr\mathcal{S}_{n_{1}}$ to the edges. So to construct a random $n_k n_{k-1}\dots n_{1}$-lift, we may assign random elements of $\mathcal{S}_{n_k} \wr \mathcal{S}_{n_{k-1}}\wr \dots \wr\mathcal{S}_{n_{1}}$ to the edges. Probabilistically, the model produced this way is the same one as taking a random $n_1$ lift, and a random $n_2$ lift of the result and so on. This construction produces the same redundancy as in the case for random $n$-lifts, and the probability of any graphical property of a random $n_k n_{k-1}\dots n_{1}$-lift does not change even if we assume that a set of edges which does not contain a cycle is assigned the identity element of $\mathcal{S}_{n_k} \wr \mathcal{S}_{n_{k-1}}\wr \dots \wr\mathcal{S}_{n_{1}}$.

There is a simple way in which the wreath product naturally arises in iterated lifts. In lifts of a graph $G$, the fiber of an edge $e$ connecting $u$ to $v$ in $G$ may connect the fibers of $u$ and $v$ through any perfect matching, but in iterated lifts it must also respect the family tree structure of the fibers of $u$ and $v$. For example, if $\tilde{G}$ is an iterated $n_2n_1$-lift of $G$, then the fiber of any vertex or edge has $n_2n_1$ elements, which may be indexed by $(i,j)$ where $i \in [1,n_1]$ and $j\in [1,n_2]$ represent the $j$th lift of the $i$th lift of the vertex or edge. If $e$ connects $u$ to $v$ in $G$, say that $i$th edge above $e$ connects the $a$th vertex over $u$ to the $b$th vertex over $v$ in the $n_1$-lift, then the $(i,j)$th edge above $e$ can only be connected to some $(a,k)$th vertex above $u$ to some $(b,l)$th vertex above $v$ in the iterated $n_2n_1$-lift. The wreath product $\mathcal{S}_{n_2}\wr \mathcal{S}_{n_1}$ is precisely the set of matchings which are restricted to respect the structure of a rooted tree in which the root has $n_1$ children, each of which have $n_2$ children. For a thorough discussion one may consult Chapter 3.3 of \cite{7}. We formally define wreath products and iterated random lifts: 
\begin{defi}
	Given two permutation groups $\mathcal{G}$ and $\mathcal{H}$ with domains T and S respectively, the wreath product of $\mathcal{G}$ and $\mathcal{H}$, denoted $\mathcal{G}\wr \mathcal{H}$, is the semi-direct product $\mathcal{G}^{|S|}\rtimes \mathcal{H}$, where the action of $h \in \mathcal{H}$ on $\mathcal{G}^{|S|}$ is defined to be $\varphi_h(g_1,g_2,\dots,g_{|S|})= \varphi(g_{h(1)},g_{h(2)}, \dots, g_{h(|S|)})$.
\end{defi}

The natural and faithful action of $\mathcal{G} \wr \mathcal{H}$ on the set $S\times T$ is defined to be: given $(\mu,\pi) \in \mathcal{G} \wr \mathcal{H}$ where $\mu \in \mathcal{G}^{|S|}$ and $\pi \in \mathcal{H}$ and $(s,t) \in S \times T$, $(\mu,\pi)(s,t) = (\pi(s),\mu_s(t))$. The wreath product of more than two groups can be obtained iteratively.

\begin{defi}
	An iterated random $n_kn_{k-1}\dots n_1$-lift, $\tilde{G}$, of a graph $G$ is constructed in the following way. First, arbitrarily orient the edges of $G$ and assign an element of  $\mathcal{S}_{n_k} \wr \mathcal{S}_{n_{k-1}}\wr \dots \wr\mathcal{S}_{n_{1}}$ to every edge of $G$ uniformly at random. The vertices of $\tilde{G}$ are $(v, (i_k,\dots, i_1))$ where $i_l \in [1,n_l]$ and $v \in V(G)$. Here, $(u,(i_k,\dots, i_1))$ is connected to $(v,(j_k,\dots,j_1))$ if and only if there is an edge $e$ connecting $u$ to $v$ in $G$, and this edge was assigned an element $\pi \in \mathcal{S}_{n_k} \wr \mathcal{S}_{n_{k-1}}\wr \dots \wr\mathcal{S}_{n_{1}}$ such that $\pi(i_k,\dots,i_1)=(j_k,\dots,j_1)$. The edges of $\tilde{G}$ are unoriented.
\end{defi}

With the usual assumption that the edges of a spanning tree are flat, we can define the walk-subgroup of an iterated random $n_k n_{k-1}\dots n_{1}$-lift of $G$ as the subgroup of $\mathcal{S}_{n_k} \wr \mathcal{S}_{n_{k-1}}\wr \dots \wr\mathcal{S}_{n_{1}}$ generated by $l$ random elements of $\mathcal{S}_{n_k} \wr \mathcal{S}_{n_{k-1}}\wr \dots \wr\mathcal{S}_{n_{1}}$ where $l = |E(G)| - |V(G)| + 1$ is the number of edges of $G$ which lie outside of the spanning tree.

\subsection{The Probability of Generating the Symmetric Group}

Given two random elements $\sigma,\tau$ of $\mathcal{S}_n$, it is natural to ask what subgroup of $\mathcal{S}_n$ they generate. It is easy to see that if $\sigma,\tau$ are both even permutations, they can only generate even permutations, and therefore cannot generate any subgroup of $\mathcal{S}_n$ bigger than $\mathcal{A}_n$. This happens with probability $\frac{1}{4}$. So we cannot hope to say that two random elements a.a.s. generate the whole of $\mathcal{S}_n$. However, in \cite{4} it is shown that the probability that two random elements of $\mathcal{S}_n$ generate $\mathcal{S}_n$ or $\mathcal{A}_n$ is at least $1 - \frac{2}{\log(\log(n))^2}$, which goes to one as $n$ increases. In \cite{3} this result is improved using facts about the classification of finite simple groups. 
\begin{prop}[Babai]The probability that two random elements of $\mathcal{S}_n$ generate $\mathcal{S}_n$ or $\mathcal{A}_n$ is $1 - \frac{1}{n} + O(\frac{1}{n^2})$.\qed\end{prop}
There is a concise summary of Babai's proof in \cite{5}: Babai's proof begins by appealing to two results of \cite{4}. The first shows that the probability that two random elements of $\mathcal{S}_n$ generate a transitive subgroup of $\mathcal{S}_n$ is $1-\frac{1}{n}+ O(\frac{1}{n^2})$. The second shows that the probability that this group is imprimitive is $\le n2^{\frac{-n}{4}}$. Babai complements these results with the following observation which relies on the classification of finite simple groups: the probability that these elements generate a primitive subgroup different from $\mathcal{S}_n$ or $\mathcal{A}_n$ is $O\left(\frac{n^{\sqrt{n}}}{n!}\right)$. It follows that the probability that two random elements of $\mathcal{S}_n$ generate a transitive subgroup of $\mathcal{S}_n$ which is not $\mathcal{S}_n$ or $\mathcal{A}_n$ is less than $O\left(n2^{\frac{-n}{4}}+\frac{n^{\sqrt{n}}}{n!}\right)$, which is $O\left(\frac{1}{n^2}\right)$.

It is easily possible to prove a general version of Babai's result for $l\ge 2$ random generators using the same proof strategy and minor modifications of arguments used by \cite{3} and \cite{4}. The following three lemmas require no new mathematical insight. \begin{lem}
	The probability that $l$ independently chosen random permutations from $\mathcal{S}_n$ fail to generate a transitive subgroup is bounded by \[\sum_{1\le r \le n/2}{n \choose r}^{1-l} \le \frac{1}{n^{l-1}} + O\left(\frac{1}{n^{l}}\right)\]
\end{lem}
\begin{proof} Replace 2 with $l$ in Lemma 1.1 of \cite{3}.
\end{proof}

\begin{lem}
	The probability that $l$ random elements generate a transitive but imprimitive subgroup of $\mathcal{S}_n$ is less than $n2^{\frac{-n(l-1)}{4}}$.
\end{lem}
\begin{proof} Replace 2 with $l$ in the arguments after Lemma 1 and in Lemma 2 of \cite{4}.
\end{proof}

\begin{lem}[Theorem 2.8 \cite{3}]
	The probability that $l$ random permutations generate a primitive group other than $\mathcal{A}_n$ or $\mathcal{S}_n$ is $O((\frac{n^{\sqrt{n}}}{n!})^{l-1})$.\qed
\end{lem}
\begin{thm}
	The probability that $l$ random elements of $\mathcal{S}_n$ generate $\mathcal{S}_n$ or $\mathcal{A}_n$ is $1 - \frac{1}{n^{l-1}} + O(\frac{1}{n^l})$. 
\end{thm}
\begin{proof} From Lemmas 3.10 and 3.11, we have that the probability that $l$ random permutations generate a transitive subgroup of $\mathcal{S}_n$, but not $\mathcal{S}_n$ or $\mathcal{A}_n$ is less than $O\left(n2^{\frac{-n(l-1)}{4}} + \left( \frac{n^{\sqrt{n}}}{n!} \right)^{l-1} \right)$ which is certainly $O\left(\frac{1}{n^l}\right)$. The result follows from Lemma 3.9.\end{proof}
In order to set up our application of this theorem, we mention the following fact from group theory: $\mathcal{S}_n$ and $\mathcal{A}_n$ act $n$-transitively and $(n-2)$-transitively on $\{1,\dots,n\}$ respectively. Keeping this mind, Theorem 3.12 can be reinterpreted as follows: the probability that $l$ random permutations generate a subgroup of $\mathcal{S}_n$ which acts at least $(n-2)$-transitively on $\{1,\dots,n\}$ is $1 - \frac{1}{n^{l-1}}+ O(\frac{1}{n^l})$.

\section{Results}
\subsection{Random Lifts}
\subsubsection{A Simple Application: Connectivity}
As an example of the utility of the walk-subgroup of random $n$-lifts, we will use it to calculate the probability of connectivity in random $n$-lifts. As we mentioned before, when considering the random $n$-lifts of $G$, we may choose a spanning tree $T$ of $G$ and assume that all edges in $T$ are flat i.e. are assigned the identity permutation. In particular, this assures that every section of the lift has a spanning tree inherited from $T$.  

\begin{thm}[Connectivity]
	Let $G$ be a simple connected graph with $l-1$ more edges than vertices (l $\ge$ 1). Then a random $n$-lift of $G$ is connected with probability $1 - \frac{1}{n^{l-1}} + O\left( \frac{1}{n^{l}} \right)$. 

\end{thm}
We make the following connection to the walk-subgroup.
\begin{prop}
	A random $n$-lift $H$ of $G$ is connected if and only if its walk-subgroup is a transitive subgroup of $\mathcal{S}_n$.
\end{prop}
\begin{proof}
Suppose that $H$ is connected. Starting at any vertex there exists a walk which can reach every other vertex. In particular, if the walk starts on the vertex $(v,i)$, it must be able to reach the vertices $(v,j)$ for all $1 \le j \le n$. The projection of the walk taking $(v,i)$ to $(v,j)$ to $G$ gives an element of the walk-subgroup, $\sigma$, such that $\sigma(i)=j$. Since this is true for all $i$ and $j$ the walk-subgroup must be a transitive subgroup of $\mathcal{S}_n$. 

Conversely suppose that the walk--subgroup is transitive. Without loss of generality, suppose a walk starts on vertex $(v,1)$. It can first walk along the lifts of the flat edges of $G$ (which form a spanning tree of every section of $H$) to cover all vertices of the form $(u,1)$. Then, it can take the walk associated with a permutation $\sigma$ such that $\sigma(1)=2$ to end up at a vertex $(a,2)$. Such an element of the walk-subgroup exists by assumption. From here it can cover all vertices of the form $(b,2)$ and continue similarly, covering the whole graph. This tells us that from every vertex there is a walk which can cover the entire graph, implying that the graph is connected.
\end{proof}
\begin{lem}
	Let $G$ be a graph with $l$ non-flat edges (or $l-1$ more edges than vertices). Then the probability that a random lift of $G$ is connected is the probability that $l$ random elements of $\mathcal{S}_n$ generate a transitive subgroup of $S_n$.
\end{lem}
\begin{proof}
	This lemma is an immediate consequence of Propositions 4.2 and 3.2.\end{proof} \begin{proof}[Proof of Theorem 4.1] This follows from Lemmas 4.3 and 3.9.\end{proof}

\subsubsection{Edge Expansion: Lower Bound}
\begin{defi}
	The isoperimetric constant or edge expansion of a graph $G$ is defined to be \[\min_{S\subset V(G), |S|\le V/2}\frac{E(S,S^c)}{|S|}\] where $E(S,S^c)$ is number of edges leaving $S$.
\end{defi}

\begin{thm}[Edge Expansion]

	Let $G$ be a simple connected graph with $l-1$ more edges than vertices (l $\ge$ 1). Then there exists a constant $\xi(G)>0$, such that a random $n$-lift of $G$ (for $n \ge 3$) has edge expansion at least $\xi(G)$, with probability $1 - \frac{1}{n^{l-1}} + O\left( \frac{1}{n^{l}} \right)$. 

\end{thm}

We make the following connection to the walk--subgroup. 
\begin{prop}
	If $H$ is a random $n$-lift of $G$ and its walk--subgroup is a $k$--transitive subgroup of $\mathcal{S}_n$ for $k \ge n/3$, then there exists a positive constant $\xi(G)$ such that $H$ has expansion at least $\xi(G)$. 
\end{prop}
\begin{proof}
Let $T$ be a subset of vertices of $H$ such that $0 < |T| \le |V(H)|/2$. For a vertex $v$ of $G$, denote the fiber over $v$ by $F_v$, and define $T_v = F_v \cap T$. Also denote $t_v = |T_v|$ and $m = \max_{v \in V(G)}t_v$. Note that $|T|< m|V(G)|$.

Fix $\varepsilon< \frac{1}{4}$. Now suppose that $t_i$ are not all of `similar size'. More precisely, suppose there exists $u$ such that $t_u < (1-\varepsilon)m$. Let $v$ be such that $t_v = m$.  We know that there are $n$ disjoint paths from $F_u$ to $F_v$ in $H$ (using the fact that $G$ is connected and the lifting property of paths), and in particular, at least $\varepsilon m$ of these paths must connect $T_v$ to a vertex outside $T_u$. Then we get \[E(T,T^c) \ge \varepsilon m = \frac{\varepsilon m |V(G)|}{|V(G)|} \ge \frac{\varepsilon |T|}{|V(G)|}\]
and so $\phi(T) \ge \frac{\varepsilon}{|V(G)|}$. Now suppose that $t_u \ge (1-\varepsilon)m$ for all $u \in V(G)$. By the choice of $\varepsilon$ it follows that $m \le 2n/3$. Consider an aribitrary $F_v$. We know that $F_v$ contains at least $n/3$ vertices not in $T$. But we know that there is an element $\sigma$ in the walk subgroup such that $|T_v \cup \sigma(T_v)| = t_v + n/3$ or $2t_v$ (in case $t_v\le n/3$). We will consider the first case, as the calculation for the second case is similar. Then there are $n/3$ indices in $T_v$ such that $\sigma(k)\notin T_v$. For all such indices $k$, the lift of the walk associated with $\sigma$ starting at $k$ contains a unique edge in $E(T,T^c)$. In particular we have \[E(T,T^c)\ge \frac{n}{3} \ge \frac{t_v}{2} \ge \frac{(1-\varepsilon)m}{2}= \frac{(1-\varepsilon)m|V(G)|}{2|V(G)|}\ge \frac{(1-\varepsilon)|T|}{2|V(G)|}\ge \frac{\varepsilon |T|}{|V(G)|}\] 
\end{proof}
\begin{proof}[Proof of Theorem 4.5] Theorem 3.12 combined with the fact the action of $\mathcal{S}_n$ is $n$-transitive and the action of $\mathcal{A}_n$ is $(n-2)$-transitive implies the lower bound on the probability in Theorem 4.5. The upper bound on the probability holds because strictly positive edge expansion implies connectedness, and the probability of connectedness in Theorem 4.1 matches the lower bound.\end{proof}

\subsubsection{$\delta$-Connectivity}

We have shown that not only are random lifts connected with high probability, but that large sets have large boundaries. We can use this to prove the following theorem about $\delta$-edge connectivity, which we will simply call $\delta$-connectivity.

\begin{thm}[$\delta$-Connectivity]

	Let $G$ be a simple connected graph with minimum degree $\delta \ge 5$. There exists $\gamma(\delta) >0$ which is strictly increasing in $\delta$ such that the probability that a random $n$-lift of $G$ is $\delta$-connected is at least $1 - O\left(\frac{1}{n^{\gamma(\delta)}}\right)$, given that $n > (\delta-1)^6|V(G)|^5$. 

\end{thm}

First we show that if we desire a non-trivial bound which works for all simple connected graphs with a fixed minimum degree, we must impose a condition on $n$ in terms of $\delta$. Consider the following example:

\begin{ex}
	The barbell graph $B_k$ consists of two cliques of k+1 vertices connected by a single edge called the bridge. This graph has minimum degree k. However, no $n$-lift of $B_k$, for $n<k$ is $k$-connected. This is because the bridge has only n copies, and cutting these n copies disconnects the graph.
\end{ex}

\FloatBarrier
\begin{figure}[!h]
\centering
\subfigure[]{
  \includegraphics[width=.55\linewidth]{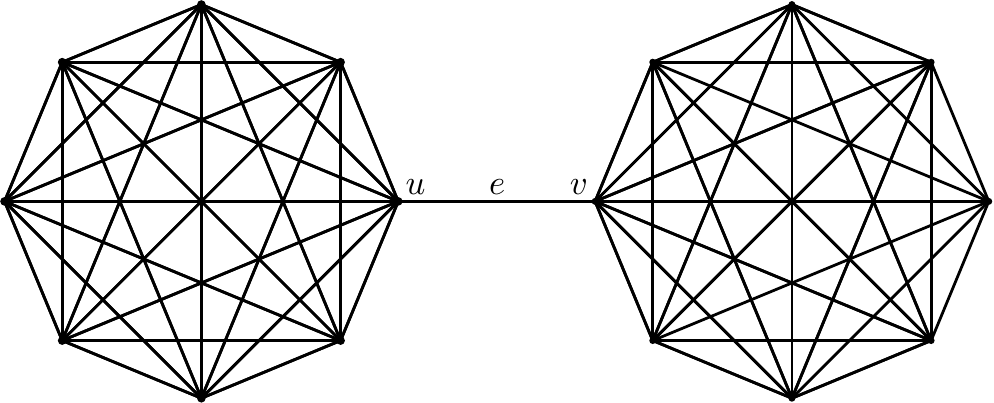}
  \label{fig:barbell}}\hspace{.5cm}
  \subfigure[]{
 \includegraphics[width=.35\linewidth]{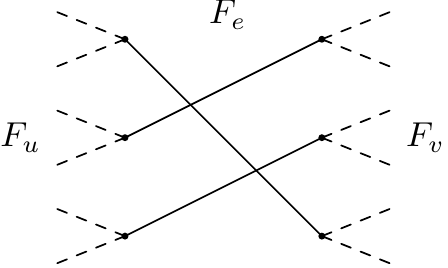}
 \label{fig:fiberbar}}
 \caption[Coverings of the Barbell Graph]{\ref{fig:barbell} is the graph $B_7$ and \ref{fig:fiberbar} shows the fiber of bridge edge $e$, denoted $F_e$, which connects the fibers of $u$ and $v$ in $B_7$, denoted $F_u$ and $F_v$ respectively. Note that no 3-lift of \ref{fig:barbell} can be 7-connected as one can simply cut every edge in $F_e$ to disconnect the graph.}
  \end{figure}
\FloatBarrier

This tells us we need $n$ to be large enough in terms of $\delta$ for $\delta$-connectivity to be possible, and the condition in our theorem, $n \ge (\delta-1)^6|V(G)|^5$, while not tight, is not a mere artifact of the proof strategy.

\begin{prop}
	Let $H$ be a random $n$-lift of $G$ where $\delta \ge 5$ is the minimum degree of $G$ and $n > (\delta-1)^6|V(G)|^5$. If the walk--subgroup of $H$ is a $\delta$-transitive subgroup of $\mathcal{S}_n$, then there exists $\gamma(\delta) >0$ which is strictly increasing in $\delta$ such that the probability that $H$ is $\delta$-connected is at least $1 - O\left(\frac{1}{n^{\gamma(\delta)}}\right)$.
\end{prop}
\begin{proof}
Let $T$ be a subset of vertices of $H$ such that $0 < |T| \le |V(H)|/2$. For a vertex $v$ of $G$, denote the fiber over $v$ by $F_v$, and define $T_v = F_v \cap T$. Also denote $t_v = |T_v|$.

First we reduce the problem to the case when the fibers over every point are roughly evenly distributed. Suppose that there exist $u,v \in G$ such that $|t_u - t_v| \ge \delta$. Then consider a path in $G$ which connects $u$ to $v$. It lifts to $n$ edge-disjoint paths in $H$ which connect $F_u$ to $F_v$, implying that $E(T,T^c) \ge \delta$. So we need only consider the case when $|x_u - x_v| \le \delta - 1$ for all $u,v \in G$. Suppose our subset of vertices contains a fiber $T_v$ such that $t_v \ge \delta$. We can assume that $|F_v \setminus T_v| \ge \delta$, because $n >>4\delta$ implies there cannot exist a single fiber such that $|F_v \setminus T_v| \le \delta$ since we are considering only sets with somewhat `balanced' fibers. Since we assumed the walk subgroup to be $\delta$-transitive, there is a loop based at $v$ in $G$ which corresponds to $\sigma$ in the walk-subgroup which lifts to $\delta$ edge disjoint paths which take $\delta$ points in $T_v$ to $\delta$ distinct points in $F_v\setminus T_v$. This implies that the boundary of such a set is at least $\delta$.   

Now we consider the only remaning case: when $t_u \le \delta-1$ for all $u \in V(G)$. These sets require the careful analysis of several cases. The first three cases show that such sets of vertices spread across a small number of fibers cannot have small boundary. The tedious case is the fourth, which (loosely) bounds the probability that the rest of such possible sets have small boundary. The argument is as follows: suppose such a set has small boundary, then it is enough to consider the case that it is a connected subgraph of $H$. In fact, we show it must be a subgraph which contains a large number of cycles, and therefore a large number of edges in $H$ which are lifts of (not necessarily distinct) non-flat edges in $G$. Since many edges in such graphs need a random permutation to take them to the correct spot (in order to complete the necessary number of cycles), they occur with low probability. Let $h$ be the number of non-empty fibers, 

\begin{itemize}[leftmargin=*]
	\item[1.] Suppose $h=1$. Since fibers are totally disconnected and the minimum degree of any vertex is $\delta$, the size of the boundary must be at least $\delta$.

	\item[2.] Suppose $2 \le h \le \delta-1$. Then we know that each vertex in this set must have at least $\delta - h + 1$ edges leaving the set. This is because each vertex can at best connect to $h-1$ other fibers (all of the fibers excluding itself). So the size of the boundary is at least $h(\delta - h + 1)$. This is minimized as a function of $h$ in the given range when $h=2$, giving us that the size of the boundary is at least $2(\delta-1) \ge \delta$. 

	\item[3.] Let $h=\delta$. In this case, each vertex has at least one edge leaving $K$, and there are at least $\delta$ vertices. So the boundary must be $\ge \delta$.

	\item[4.] Now let $h > \delta$. We may assume that such a set $K$ (of size $k$), is a connected subgraph of $H$, since disconnected subgraphs have a boundary greater than or equal to the boundary of any of the components. We first show that any $K$ with boundary $< \delta$ must have at least $1.3k$ edges more than vertices. The vertices of $K$ have minimum degree $\delta$, implying that the total degree of $K$ is at least $k\delta$. Since $K$ is connected it has a spanning tree with $k-1$ edges, which contributes $2k-2$ to the total degree of $K$. Of the remaining $k\delta - 2k + 2$ total degree, at least $(k-1)\delta - 2k + 3$ must be accounted for by edges that connect back into the graph. This is because at most $\delta-1$ go outside $K$ by assumption. By eliminating the double counting of edges that stay within $K$, the total number of non-spanning tree edges in $K$ is at least \[\frac{(k-1)\delta - 2k + 3}{2} \ge \frac{5k-5-2k +3}{2} \ge 1.5k - 1 \ge 1.3k\] where we use that $k > \delta \ge 5$. Note that this quantity strictly increases with $\delta$. We continue the rest of the proof for $\delta=5$ which is the worst case covered by our theorem, and it is clear that larger $\delta$ will result in better versions of the bounds to follow. We argue that since $K$ has at least $1.3k$ edges in excess of a spanning tree, it must have at least $1.3k$ edges which are lifts of (not necessarily distinct) non-flat edges in $G$. For the sake of contradiction suppose that $K$ has less than $1.3k$ non-flat edges. Then upon deleting them, we are left with lifts of flat edges only, but more edges than in a spanning tree of $K$. That means that we must have at least one cycle in $K$, which must come from a cycle in $G$. But a cycle in $G$ must contain at least one non-flat edge, and therefore $K$ must still contain at least one edge which is a lift of a non-flat edge.

Now suppose that $m$ of these edges lie above a single edge in $G$ (note that $m \le \delta-1$). The probability that a random permutation takes them to the correct points in their destination fiber to keep them within the subgraph is less than \[\frac{\delta-1}{n}\times \frac{\delta-2}{n-1} \times \dots \times \frac{\delta-m-1}{n-m} \le \left( \frac{\delta-1}{n} \right)^m\] where the inequality follows since $n >> \delta$. Now notice that lifts of different non-flat edges of $G$ are independent, which combined with the previous observation gives us that the probability that the necessary $1.3k$ edges stay within the subgraph is less than $\left( \frac{\delta-1}{n} \right)^{1.3k}$.

This shows us that the probability that a connected subgraph of $k$ vertices has a small boundary is very small. The total number of such subgraphs is certainly less than the number of sets of vertices of size $k$, ${n|V(G)| \choose k} = O(|V(G)|^kn^k)$. So by the union bound, the probability that any such subgraph of size $k$ exists is certainly on the order of \begin{equation}\label{equation:replace}|V(G)|^kn^k\left(\frac{\delta-1}{n}\right)^{1.3k} =  |V(G)|^{k}\frac{(\delta-1)^{1.3k}}{n^{.3k}} < \frac{1}{n^{.05k}}  \end{equation} where the second inequality uses the fact that $n > (\delta-1)^6|V(G)|^5$.
\end{itemize}
Finally through union bound, the probability that any bad subgraph of any size exists is less than \[\sum_{i=\delta}^{(\delta-1)|V(G)|}\frac{1}{n^{.05i}} < (\delta-1)|V(G)|\frac{1}{n^{.05\delta}} < \frac{1}{n^{.05}}\] where we again use the fact that $n > (\delta-1)^6|V(G)|^5$. This completes the proof.
\end{proof}

\begin{proof}[Proof of Theorem 4.7] Since for $\delta \ge 5$ the number of non-flat edges is much greater than $\delta$, then Theorem 3.12 and Proposition 4.9 imply Theorem 4.7 through the union bound. \end{proof}

The following theorem shows $\delta$-connectivity in $n$-lifts of families of graphs where $\delta$ grows slowly enough as a function of the degree of the lift and the number of vertices. 
\begin{thm}
	There exists $\gamma' >0$ such that for all $5 \le \delta(n,k) \le O\left( \frac{n^{\gamma'}}{k} \right)$, a random $n$-lift of any connected simple graph with on $k$ vertices with minimum degree $\delta(n,k)$ is a.a.s. $\delta$-connected.
\end{thm}
\begin{proof}
	The proof is the same as Theorem 4.7 using Proposition 4.9: use $\gamma' = .19$ and follow the proof from (\ref{equation:replace}) to show a $O\left(\frac{1}{n^{.06}}\right) = o_n(1)$ probability that $\delta$-connectivity fails.
\end{proof}

\subsubsection{Iterated Random Lifts}
All our results can be generalized to iterated random lifts by simply viewing iterated random lifts as a sequence of lifts. We show an analogue of Theorem 4.1 only, but the results about edge expansion and $\delta$-connectivity hold as well.

\begin{thm}
	Let $G$ be a simple connected graph with $l-1$ more edges than vertices (l $\ge$ 1), then an iterated random $n_k\dots n_1$-lift is connected with probability 
	\[\left(1 -  \frac{1}{n_1^{l-1}} + O\left(\frac{1}{n_1^{l}}\right)\right)\prod_{i=2}^{k} \left(1 -  \frac{1}{n_i^{(l-1)(\prod_{j=1}^{i-1}n_j)}} + O\left(\frac{1}{n_i^{(l-1)(\prod_{j=1}^{i-1}n_j)+1}}\right)\right)\]
\end{thm}
\begin{proof}
	Following our discussion of iterated random lifts: an iterated random $n_k\dots n_1$-lift is a random $n_k$-lift of an iterated random $n_{k-1}\dots n_1$-lift and so on, beginning with a random $n_1$-lift of $G$. By independence, the probability that an iterated random $n_k\dots n_1$-lift is connected is just the product of the probabilities that each graph in its iterated construction is connected. We can calculate this probability for each graph in the iterated construction using Theorem 4.1. Since $G$ has $l-1$ more edges than vertices, we can easily calculate that an iterated $n_i \dots n_1$-lift of $G$ has $(l-1)n_1\dots n_i$ more edges than vertices, and the result follows.\end{proof}

\subsection{Homotopy Invariants in Random Lifts}
The results of Theorems 4.1 and 4.5 apply to random lifts of graphs which are not simple, that assumption was needed only in the proofs pertaining to $\delta$-connectivity but adopted throughout for continuity. Lifts of graphs are topological covering spaces and it is well known that covering spaces of homotopy equivalent spaces have some similarities. We show that there exist graphical properties of random lifts whose probability only depends on the homotopy type of their base graph, that is, properties whose probabilities are homotopy  invariant.

It is not straightforward in general to determine whether two spaces are homotopy equivalent, however, the situation is easy for graphs. It is easy to show that two connected graphs are homotopy equivalent if and only if the number of edges minus the number of vertices is the same. Using this we show that the probability that a random lift of $G$ is connected or has edge expansion bounded below by $\xi(G)$ (from Theorem 4.5) only depends on the homotopy type of $G$. While random models of covering spaces have only been studied for graphs so far, we expect such invariants to exist for any model of random covering spaces of topological spaces. Though we only state the following results for random lifts, similar properties of iterated random lifts can be shown by the same method.

\begin{thm}
	Let $G$ and $H$ be connected graphs. Then the probabilities that a random $n$-lift of $G $ and a random $n$-lift of $H$ are connected are equal if and only if $G$ is homotopy equivalent to $H$. 
\end{thm}
\begin{proof}
The probability that these lifts are connected are simply the probabilities that their walk-subgroups are subgroups of $\mathcal{S}_n$ which act transitively on $\{1,\dots, n\}$. These probabilities only depend on the number of generators for their respective walk-subgroups, which is the same for random $n$-lifts of homotopy equivalent graphs.
\end{proof}
\begin{thm}
	Let $G$ and $H$ be connected graphs. Then the probabilities that a random $n$-lift of $G $ and a random $n$-lift of $H$ have edge expansion bounded below by $\xi(G)$ and $\xi(H)$ respectively are equal if and only if $G$ is homotopy equivalent to $H$. 
\end{thm}
\begin{proof}
Similar to the previous proof: this fact is implied by transitivity properties of the walk-subgroup which depend only on the number of generators, which is the same for random $n$-lifts of homotopy equivalent graphs.
\end{proof}

We will provide another perspective on homotopy invariants of random lifts. Any graph $G$ with $d = |E(G)| - |V(G)| + 1$ is homotopy equivalent to the bouquet of $d$-circles, $C_d$ which consists of a single vertex with $d$ loops. Lifts of $C_d$ are well studied, and also known as unions of permutations or random $2d$-regular multigraphs. A random lift of $C_d$ or a random $2d$-regular multigraph on $n$ vertices is obtained by choosing $d$ permutations $\sigma_1,\dots, \sigma_d$ independently and randomly from $\mathcal{S}_n$, and adding the edges $(j, \sigma_i(j))$ for all $i$ to the $n$ (initially isolated) vertices. Loops count as incoming and outgoing edges in such graphs. 
\begin{prop}
	A random 2d-regular multigraph $H$ is connected if and only if its walk-subgroup is a transitive subgroup of $S_n$. 
\end{prop}
\begin{proof}
First note that if $H$ is constructed using the permutations $\sigma_1,\dots, \sigma_d$ then the walk-subgroup of $H$ is indeed the subgroup generated by the $\sigma_1,\dots,\sigma_d$. This is clear since walks exist which traverse the loops of $C_d$ in all possible orders. Now suppose the walk-subgroup is transitive. Then to get to any $u$ from any $v$ in $H$, simply take the lift of a walk associated with $\sigma \in S_n$ such that $\sigma(u) = v$. Conversely, if $H$ is connected then for any $u$ and $v$, there is a walk from $u$ to $v$. The walk-product of the projection of such a walk gives element of the walk-subgroup such that $\sigma(u)=v$. 
\end{proof}
\begin{thm}

	A random 2d-regular multigraph is connected with probability $1 - \frac{1}{n^{d-1}} + O\left( \frac{1}{n^{d}} \right)$.

\end{thm}
\begin{proof}
Use Proposition 4.14 and Lemma 3.9.
\end{proof}

\begin{prop}
	The edge expansion of a random 2d-regular multigraph whose walk-subgroup is a k-transitive subgroup of $S_n$ for $k\ge n/2$ is at least 1.
\end{prop}
\begin{proof}
Suppose the walk-subgroup is $n/2$-transitive, then for any set $T$ of size $\le n/2$, we may use the lift of the walk associated with element of the walk subgroup $\sigma$ such that $\sigma(t) \not= t$ for any $t \in T$ to show that there must be at least $\min(|T|, n/2)$ edges leaving it. In either case, the edge expansion has to be greater than one.\end{proof}
\begin{thm}
A random 2d-regular multigraph has edge expansion at least $1$ with probability $1 - \frac{1}{n^{d-1}} + O\left( \frac{1}{n^{d}} \right)$.
\end{thm}
\begin{proof}
Use Proposition 4.16 and Theorem 3.12.
\end{proof}

In particular these results show that the homotopy invariants of random lifts can simply be studied as properties which hold with the same probability for random lifts of $C_d$ and $C_{d'}$ if and only if $d = d'$.

Though we have focused on Amit and Linial's model of random lifts obtained by uniform probability assignments from $\mathcal{S}_n$ to edges, it is possible to construct [restricted] models of random lifts through assignments from any group using any distribution.  Even in this general setting it is true that connectivity is a homotopy invariant in the manner described above.

\subsection{The Probability of Generating a Transitive Subgroup of $\mathcal{S}_{n_k}\wr \dots \wr \mathcal{S}_{n_1}$}
We use Theorem 4.11 to calculate the probability that $l$ random elements of $\mathcal{S}_{n_1} \wr \dots \wr \mathcal{S}_{n_k}$ produce a subgroup which acts transitively on $N_k \times \dots \times N_1$ where $N_i$ is [1,$n_i$]. This provides a generalization of Lemma 3.9 to wreath products of symmetric groups.

\begin{thm}
	The probability that l independently chosen permutations from  $\mathcal{S}_{n_k} \wr \dots \wr \mathcal{S}_{n_{1}}$ generate a subgroup of  $\mathcal{S}_{n_k} \wr \dots \wr \mathcal{S}_{n_{1}}$  which acts transitively on $N_k\times \dots \times N_1$ is
\[\left(1 -  \frac{1}{n_1^{l-1}} + O\left(\frac{1}{n_1^{l}}\right)\right)\prod_{i=2}^{k} \left(1 -  \frac{1}{n_i^{(l-1)(\prod_{j=1}^{i-1}n_j)}} + O\left(\frac{1}{n_i^{(l-1)(\prod_{j=1}^{i-1}n_j)+1}}\right)\right)\]
\end{thm}

\begin{proof} Notice that Proposition 4.2 can be generalized for any permutation group acting transitively on its domain and the theorem follows from Theorem 4.11.\end{proof}

\section{Discussion}
Most known results from literature about random lifts give probability estimates of the form $1-o_n(1)$ where $n$ is the degree of the lift; we give actual rates of convergence to one using ideas from group theory. We focused on connectivity properties, but we think that studying random lifts through random permutations or random permutation groups could have wider applications in proving new results and improving known ones.

It would be interesting to further investigate homotopy invariants in random lifts, and we expect such homotopy based inheritance to exist for any random construction of topological covering spaces.

\begin{centering}
\textbf{Acknowledgements}\\
\end{centering}
I would like to thank Thomas Sauerwald, Alexander Makelov and Luke Kweku Abraham for their helpful comments and discussions, and St John's College for supporting me through a Benefactor's Scholarship for Research.

% vim: set spell:
\end{document}